\documentclass[aps,prl,twocolumn, superscriptaddress,groupedaddress,nofootinbib]{revtex4-1}
\usepackage[dvipsnames]{xcolor}
\usepackage[utf8]{inputenc}
\usepackage[breaklinks]{hyperref}
\usepackage{amsmath}
\usepackage{tikz}
\usepackage{extpfeil}
\usepackage{amsfonts}
\usepackage{amssymb} 
\usepackage{bm}
\usepackage{graphicx}
\usepackage{array}
\usepackage{lipsum}
\usepackage{float}
\usepackage{multirow}
\usepackage{hhline}
\usepackage{tabularx}
\usepackage{subfigure}
\usepackage{graphicx}
\usepackage{afterpage}
\usepackage{longtable}
\usepackage{epstopdf}
\usepackage{MnSymbol}

\usepackage{amsthm}
\newtheorem{prop}{Proposition}

\newcommand{\bZ}{\mathbb{Z}}

\newcommand{\vev}[1]{\langle #1 \rangle}

\newcommand{\sdoc}{S_{\Delta_1^\circ}}
\newcommand{\sdtc}{S_{\Delta_2^\circ}}
\newcommand{\doc}{{\Delta_1^\circ}}
\newcommand{\dtc}{{\Delta_2^\circ}}
\newcommand{\textin}{\,\, \text{in} \,\,}

\def\ge{E}
\def\gso{SO}
\def\gsu{SU}
\def\gsp{Sp}
\def\gf{F}
\def\gg{G}

\begin{document}
\title{On Algorithmic Universality in F-theory Compactifications}
\author{James Halverson, Cody Long, and Benjamin Sung}
\affiliation{Department of Physics, Northeastern University \\ Boston, MA 02115-5000 USA} 

\date{\today}

\begin{abstract}
We study universality of geometric gauge sectors in the string landscape
in the context of F-theory compactifications.
 A finite time construction algorithm is presented for $\frac43 \times 2.96 \times 10^{755}$ F-theory geometries that are
connected by a network of topological transitions in a
connected moduli space. High probability geometric
assumptions uncover universal structures in the ensemble without
explicitly constructing it. For example, non-Higgsable clusters
of seven-branes with intricate gauge sectors occur with
probability above $1-1.01\times 10^{-755}$,  and
the geometric gauge group rank is above $160$ with probability
$.999995$. In the latter case there are at least $10$ $E_8$
factors, the structure of which fixes the gauge groups
on certain nearby seven-branes. Visible sectors may arise from $E_6$ or $SU(3)$ seven-branes, which occur
in certain random samples with probability $\simeq 1/200$.
\end{abstract}

\maketitle

\noindent{\bf I. Introduction.} 
String theory is a consistent theory of quantum gravity that naturally
gives rise to interesting gauge and cosmological sectors. As
such, it is a promising candidate for a unified theory. However, there is a vast 
landscape of four-dimensional metastable vacua that
may realize different physics, making predictions difficult.

A possible way forward, as in many areas of physics, is to demonstrate
universality in large ensembles. For string vacua, such as the oft-quoted
$O(10^{500})$ type IIb flux vacua
\cite{Bousso:2000xa,*Ashok:2003gk,*Denef:2004ze}, studying universality via
explicit construction is complex and impractical
\cite{Denef:2006ad,*Cvetic:2010ky}. However, it may be possible to derive
universality from a precise construction algorithm, rather than from the
constructed ensemble. We refer to this as algorithmic universality, and find it a
promising way forward in the string landscape.

We present such an algorithm in the context of $4d$ F-theory \cite{Vafa:1996xn,*Morrison:1996pp,*Morrison:1996pp}
compactifications. The ensemble is a collection of
$4/3\times2.96\times 10^{755}$ six-manifolds, perhaps the largest set
of string geometries to date, that serve as the extra spatial dimensions. Their topological structure determines the $4d$ gauge group that arises 
geometrically from
configurations of seven-branes that form a network of so-called non-Higgsable
clusters (NHC) \cite{Morrison:2012np}. We establish that non-Higgsable clusters arise with
probability above $1-1.01\times 10^{-755}$ in this ensemble, and demonstrate that a rich minimal gauge structure arises with high probability. We also
present results from random sampling that are potentially relevant for
visible sectors.

A number of recent results suggest that NHC are 
important in the $4d$ F-theory landscape. They exist for generic vacuum expectation values of
 scalar fields (complex structure moduli), and therefore
gauge symmetry does not require stabilization on subloci in 
moduli space~\cite{Grassi:2014zxa}, which can have high codimension \cite{Braun:2014xka,*Watari:2015ysa,*Halverson:2016tve}. Standard model structures may arise naturally \cite{Grassi:2014zxa}, strong coupling is generic \cite{Halverson:2016vwx}, and $4d$ NHC may exhibit features \cite{Morrison:2014lca} (such as loops and branches) not present
in $6d$. NHC arise in the geometry with the largest number of flux
vacua \cite{Taylor:2015xtz}, and universally in known ensembles 
\cite{Halverson:2015jua,*Taylor:2015ppa}  closely related to ours.
$6d$ NHC have been studied extensively \cite{Morrison:2012np,Morrison:2012js,*Taylor:2012dr,*Morrison:2014era,*Martini:2014iza,*Johnson:2014xpa,*Taylor:2015isa}.

In Sec. II we review non-Higgsable clusters. In Sec. III we present our ensemble.
In Sec. IV we exhibit universality. In Sec. V we discuss our results.

\vspace{.2cm}
\noindent{\bf II. Seven-Branes and Non-Higgsable Clusters.}
A $4d$ F-theory geometry is a Calabi-Yau elliptic fibration
$X$ over six extra spatial dimensions  described by a complex threefold base space $B$ defined by the equation
\begin{equation}
y^2=x^3+f x + g
\end{equation}
where $f$ and $g$ are polynomials in the coordinates
of $B$;  technically, $f\in \Gamma(\mathcal{O}(-4K_B))$, $g\in \Gamma(\mathcal{O}(-6K_B))$, with $K_B$ the canonical class.
Seven-branes are localized on the discriminant locus
$\Delta=4f^3+27g^2=0\subset B$.

Upon compactification the gauge group structure of seven-branes gives rise
to four-dimensional gauge sectors. It is controlled by $f$ and
$g$, and for a typical $B$ the most general $f, g$ take the
form $f=\tilde f \prod_i x_i^{l_i}$, $g=\tilde g \prod_i x_i^{m_i}$,
so
\begin{equation}
\Delta = \tilde \Delta\,\,  \prod_i x_i^{\text{min}(3l_i,2m_i)}=: \tilde \Delta\,\,  \prod_i x_i^{n_i},
\end{equation}
and therefore $f$, $g$, and $\Delta$ vanish along $x_i=0$ to 
$mult_{x_i=0}(f,g,\Delta)=(l_i,m_i,n_i).$ This seven-brane carries
a gauge group $G_i$ given in Table \ref{tab:gauge}
according to the Kodaira classification. In some cases
further geometric data is necessary to uniquely 
specify $G_i$ (see e.g. 
\cite{Halverson:2015jua} for conditions) but this data
always exists for fixed $B$. 
For generic $f$ and $g$ a seven-brane on $x_i=0$
requires $(l_i,m_i)\geq(1,1)$.

Such a seven-brane is called a geometrically non-Higgsable
seven-brane (NH7) because it carries a gauge group that cannot
be removed by deforming $f$ or $g$. A NH7 may have geometric gauge group
\begin{equation}
G\in \{E_8,E_7,E_6,F_4,SO(8),SO(7),G_2,SU(3),SU(2)\}, \nonumber
\end{equation} 
which could be broken by fluxes. We assume fluxes can be turned on in a large fraction
of our geometries.
A typical base $B$, as
we will show in the strongest generality to date, has many non-Higgsable seven-branes
that often intersect in pairs, giving rise to jointly charged
matter. This is a geometrically
non-Higgsable cluster (NHC). For brevity, we henceforth drop 
geometric and geometrically.

\begin{table}[t]
\scalebox{.9}{\begin{tabular}{|c|c|c|c|c|c|}
\hline
$F_i$ & $l_i$ & $m_i$ & $n_i$ & Sing. & $G_i$ \\ \hline \hline
$I_0$&$\geq $ 0 & $\geq $ 0 & 0 & none & none \\
$I_n$ &0 & 0 & $n \geq 2$ & $A_{n-1}$ & $\gsu(n)$  or $\gsp(\lfloor
n/2\rfloor)$\\
$II$ & $\geq 1$ & 1 & 2 & none & none \\
$III$ &1 & $\geq 2$ &3 & $A_1$ & $\gsu(2)$ \\
$IV$ & $\geq 2$ & 2 & 4 & $A_2$ & $\gsu(3)$  or $\gsu(2)$\\
$I_0^*$&
$\geq 2$ & $\geq 3$ & $6$ &$D_{4}$ & $\gso(8)$ or $\gso(7)$ or $\gg_2$ \\
$I_n^*$&
2 & 3 & $n \geq 7$ & $D_{n -2}$ & $\gso(2n-4)$  or $\gso(2n -5)$ \\
$IV^*$& $\geq 3$ & 4 & 8 & $E_6$ & $\ge_6$  or $\gf_4$\\
$III^*$&3 & $\geq 5$ & 9 & $E_7$ & $\ge_7$ \\
$II^*$& $\geq 4$ & 5 & 10 & $E_8$ & $\ge_8$ \\ \hline
\end{tabular}}
\caption{Kodaira fiber $F_i$, singularity, and gauge group $G_i$ on
the seven-brane at $x_i=0$ for given $l_i$, $m_i$, and $n_i$.}
\label{tab:gauge}
\end{table}

\vspace{.2cm}
\noindent{\bf III. Large Landscapes of Geometries from Trees.}

We now introduce our construction, which utilizes building blocks in toric
varieties that we call trees to systematically build up F-theory geometries.
After describing the geometric setup and defining terms that simplify the
discussion, we will present a criterion, classify all trees satisfying it, and
build the F-theory geometries.

Our construction begins with a smooth weak-Fano toric threefold $B_i$, and
then builds structure on top of it. Each geometry $B_i$ is  determined by a
fine regular star triangulations (FRST) of one of the $4319$  $3d$ reflexive
polytopes \cite{Kreuzer:1998vb}; there are an estimated
 $O(10^{15})$ such geometries \cite{Halverson:2016tve}. The $2d$
faces of the $3d$ polytope are known as facets, and a triangulated polytope
will have triangulated facets.  Such $B_i$ do not support NHC; see~\cite{Halverson:2016tve}.

Consider such a $B_i$ determined by an FRST of a $3d$ reflexive polytope
$\Delta^\circ$, a triangulated facet $F$ in $\Delta^\circ$, and an edge
between two points $v_1$ and $v_2$ in $F$ with
associated homogeneous coordinates $x_1$ and $x_2$. Since $v_{1,2}$ are connected
by an edge, $x_1=x_2=0$ defines a Riemann surface (algebraic curve)
in $B_i$, which can be ``blown up" using a new ray $v_e=v_1+v_2$ and subdividing
cones using standard toric techniques. This is a topological transition that introduces
a new (``exceptional") divisor $e=0$ in $B$, where $e$ is the coordinate associated to $v_e$.
This
process can be iterated, for example blowing up along $e=x_1=0$, which would
add a new ray $v_e+v_1=2v_1+v_2$.

After a number of iterations the associated toric variety will have a
collection of exceptional divisors with associated rays $v_{e_i}=a_i v_1 + b_i
v_2$, which will appear to have formed a \emph{tree} above the \emph{ground}
that connects $v_1$ and $v_2$ in $F$. Each $v_{e_i}$ is a $\emph{leaf}$ with
\emph{height} $h_{e_i}=a_i+b_i$, and  we will refer to trees built on edges
within $F$ as $\emph{edge trees}$. The height of a tree is the height of its
highest leaf. As an example, $\{v_1+v_2,2v_1+v_2,v_1+2v_2\}$ appears as
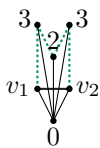
\begin{center} 
\begin{tikzpicture}[scale=.85]
 \draw[thick,color=Black] (.25,.5) --
(-.25,.5);\draw[thick,dash pattern={on 1pt off 1pt},color=ForestGreen] (-.25,.5)--(-.25,1.5)--(0,1)--(.25,1.5)--(.25,.5); \fill (0,0) circle (.5mm); \fill (.25,.5) circle (.5mm); \fill
(-.25,.5) circle (.5mm); \fill (0,1) circle (.5mm); \fill (-.25,1.5) circle
(.5mm); \fill (.25,1.5) circle (.5mm); \draw (0,0) -- (.25,.5); \draw (0,0) --
(-.25,.5); \draw (0,0) -- (0,1); \draw (0,0) -- (-.25,1.5); \draw (0,0) --
(.25,1.5); \node at (-.55,.5) {$v_1$}; \node at (.55,.5) {$v_2$}; \node at
(0,-.25) {$0$}; \node at (0,1.25) {$2$}; \node at (-.45,1.6) {$3$}; \node at
(.45,1.6) {$3$}; 
\end{tikzpicture} 
\end{center} where the $v_1$ to $v_2$ line is the
edge (ground) in $F$, dashed green lines are above the ground, $0$ is the origin of $\Delta^\circ$, and the new rays
are labeled by height.

Similarly, one can also build \emph{face trees} by beginning with a face on
$F$, with vertices $v_1, v_2, v_3$ associated to $x_1, x_2, x_3$. Adding $v_e=v_1+v_2+v_3$ and subdividing
appropriately blows up the point $x_1=x_2=x_3=0$  and produces a new toric
variety. Again such blowups can be iterated. This process builds a collection
of leaves $v_{e_i}=a_i v_1 + b_i v_2 + c_i v_3$ with $a_i,b_i,c_i > 0$ of
height $h_{e_i}=a_i+b_i+c_i$ that comprise a \emph{face tree}. Face trees are
built above the interior of the face due to the strict inequality in the
definition. Note if one leaf coefficient was zero the associated leaf would be
above an edge of the face, not above the face interior.

Geometries can be systematically constructed by adding a face tree to each face in each triangulated facet of $\Delta^\circ$, and then an edge tree to each edge. The associated smooth toric
threefold $B$ has a collection of rays $v$, each of which can be written $v=av_1+bv_2+cv_3$ with $v_i$ $3d$ cone vertices in $B_i$. If $(a,b,c)=(1,0,0)$ or some
permutation thereof, $v\in \Delta^\circ$ and this height $h_v=1$ ``leaf'' is more
appropriately a root, since it is on the ground.

A natural question
in systematically building up geometries is whether there is
a maximal tree height. For a toric variety $B$ to be an allowed
F-theory base it must not have any so-called $(4,6)$ divisors (see Appendix), which
we ensure by a simple height criterion proven in Prop. \ref{prop:heightcrit}:
\begin{center}
If $h_v\leq 6$ for all leaves $v\in B$, \\ then there are no $(4,6)$
divisors.
\end{center}
This condition is simple and sufficient, but not necessary,
for the absence of $(4,6)$ divisors. Nevertheless, it
will allow us to build a large class of geometries.

\vspace{.5cm}
The task is now clear: we must systematically build  all
topologically distinct edge trees and face trees of height 
$\leq 6$. Since the combinatorics are daunting, let us
exemplify the problem for $h\leq 3$ trees. Viewing the
facet head on, an edge in $F$ appears as 
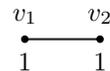
\begin{center}
\begin{tikzpicture}
\draw[thick,color=Black] (0,0) -- (1,0);
\fill (0,0) circle (.5mm);
\fill (1,0) circle (.5mm);
\node at (0,.3) {$v_1$};
\node at (1,.3) {$v_2$};
\node at (0,-.3) {$1$};
\node at (1,-.3) {$1$};
\end{tikzpicture}
\end{center}
with the vertices and their heights labeled.
Adding $v_1+v_2$ subdivides the
edge, and further subdivision gives
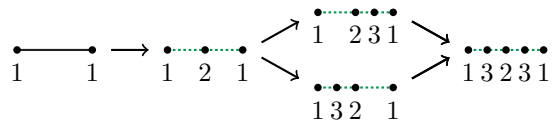
\begin{center}
\begin{tikzpicture}
\draw[thick,color=Black] (0,0) -- (1,0);
\fill (0,0) circle (.5mm);
\fill (1,0) circle (.5mm);
\node at (0,-.3) {$1$};
\node at (1,-.3) {$1$};
\draw[thick,->] (1.25,0) -- (1.75,0);
\draw[thick,dash pattern={on 1pt off 1pt},color=ForestGreen] (2,0) -- (3,0);
\fill (2,0) circle (.5mm);
\node at (2,-.3) {$1$};
\fill (2.5,0) circle (.5mm);
\node at (2.5,-.3) {$2$};
\fill (3,0) circle (.5mm);
\node at (3,-.3) {$1$};
\draw[thick,->] (3.25,.1) -- (3.75,.38);
\draw[thick,->] (3.25,-.1) -- (3.75,-.38);
\draw[thick,->] (5.25,.38) -- (5.75,.1);
\draw[thick,->] (5.25,-.38) -- (5.75,-.1);
\draw[thick,dash pattern={on 1pt off 1pt},color=ForestGreen] (4,.5) -- (5,.5);
\fill (4,.0+.5) circle (.5mm);
\node at (4,-.3+.5) {$1$};
\fill (4.5,0+.5) circle (.5mm);
\node at (4.5,-.3+.5) {$2$};
\fill (4.75,0+.5) circle (.5mm);
\node at (4.75,-.3+.5) {$3$};
\fill (5,0+.5) circle (.5mm);
\node at (5,-.3+.5) {$1$};
\draw[thick,dash pattern={on 1pt off 1pt},color=ForestGreen] (4,-.5) -- (5,-.5);
\fill (4,.0-.5) circle (.5mm);
\node at (4,-.3-.5) {$1$};
\fill (4.5,0-.5) circle (.5mm);
\node at (4.5,-.3-.5) {$2$};
\fill (4.25,0-.5) circle (.5mm);
\node at (4.25,-.3-.5) {$3$};
\fill (5,0-.5) circle (.5mm);
\node at (5,-.3-.5) {$1$};
\draw[thick,->] (5.25,.38) -- (5.75,.1);
\draw[thick,->] (5.25,-.38) -- (5.75,-.1);
\draw[thick,dash pattern={on 1pt off 1pt},color=ForestGreen] (6,0) -- (7,0);
\fill (6,.0) circle (.5mm);
\node at (6,-.3) {$1$};
\fill (6.25,0) circle (.5mm);
\node at (6.25,-.3) {$3$};
\fill (6.5,0) circle (.5mm);
\node at (6.5,-.3) {$2$};
\fill (6.75,0) circle (.5mm);
\node at (6.75,-.3) {$3$};
\fill (7,0) circle (.5mm);
\node at (7,-.3) {$1$};
\end{tikzpicture}
\end{center}
where we have dropped the vertex labels and kept the heights.
The trees emerge out of the page, but visualization is made easier
by projecting on to the edge; the right-most tree is the
one previously presented vertically. There are five
edge trees with height $\leq 3$. Similarly,
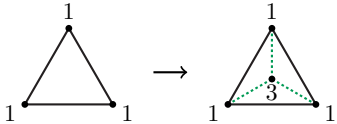
\begin{center}
\begin{tikzpicture}[scale=0.9, every node/.style={scale=0.9}]
\draw[thick,color=Black] (90:.75) -- (90+120:.75) -- (90+120+120:.75) -- cycle;
\fill (90:.75) circle (.5mm);
\fill (90+120:.75) circle (.5mm);
\fill (90+240:.75) circle (.5mm);
\node at (90:1) {$1$}; \node at (90+120:1) {$1$}; \node at (90+240:1) {$1$};
\draw[thick,->] (1.25,.1) -- (1.75,.1);
\begin{scope}[xshift=3cm]
\draw[thick,dash pattern={on 1pt off 1pt},color=ForestGreen] (90:.75) -- (0,0);
\draw[thick,dash pattern={on 1pt off 1pt},color=ForestGreen] (90+120:.75) -- (0,0);
\draw[thick,dash pattern={on 1pt off 1pt},color=ForestGreen] (90+240:.75) -- (0,0);
\fill (0,0) circle (.5mm);
\draw[thick,color=Black] (90:.75) -- (90+120:.75) -- (90+120+120:.75) -- cycle;
\fill (90:.75) circle (.5mm);
\fill (90+120:.75) circle (.5mm);
\fill (90+240:.75) circle (.5mm);
\node at (90:1) {$1$}; \node at (90+120:1) {$1$}; \node at (90+240:1) {$1$};
\node at (0,-.2) {$3$};
\end{scope}
\end{tikzpicture}
\end{center}
shows that there are $2$ face trees of height $\leq 3$. Here we have denoted the new edges by green lines since they do
not sit in the facet. With our definitions, edge trees are built above an edge
in the facet, whereas higher leaves in face trees may be built on new edges
that do not sit in the facet. For example, a height $4$ leaf could be added
on any of the green lines above. A (tedious) straightforward calculation shows that the number of
edge or face trees with $h \leq N$ grows rapidly with $N$, as in
Table \ref{tab:numedgefacetreeandprob}.
\begin{table}
\begin{tabular}{|c|c|c|}
\hline
$N$ & \# Edge Trees & \# Face Trees \\ \hline
$3$ & $5$ & $2$\\
$4$ & $10$ & $17$\\
$5$ & $50$ & $4231$ \\
$6$ & $82$ & $41,873,645$\\ \hline
\end{tabular} \hspace{1cm}
\begin{tabular}{|c|c|}
\hline
$h_v$ & Probability \\ \hline
$3$ & $.99999998$ \\
$4$ & $.999995$ \\
$5$ & $.999997$  \\
$6$ & $.999899$ \\ \hline
\end{tabular}
\caption{\emph{Left:} The number of edge trees and face trees with height $h\leq N$.
\emph{Right:} The probability that a face tree with $h\leq 6$ has a leaf $v$ with
a given height $h_v$.}
\label{tab:numedgefacetreeandprob}
\end{table}

\vspace{.5cm}
Having classified the number of $h\leq 6$ face trees and edge trees, we now give a lower bound for
the number of F-theory geometries that arise from building trees on an FRST of 
$\Delta^\circ$, denoted $\mathcal{T}(\Delta^\circ)$. 
We construct an ensemble $S_{\Delta^\circ}$ of geometries by systematically putting $h\leq 6$ face trees on all
faces $\tilde F$ of $\mathcal{T}(\Delta^\circ)$ and then putting $h\leq 6$ edge trees on
all edges $\tilde E$ of $\mathcal{T}(\Delta^\circ)$. Using Table \ref{tab:numedgefacetreeandprob}, the size of $S_{\Delta^\circ}$ is 
\begin{equation}
|S_{\Delta^\circ}| = 82^{\# \tilde E \, \text{on} \, \mathcal{T}(\Delta^\circ)} \times (4.19\times 10^6)^{\# \tilde F \, \text{on} \, \mathcal{T}(\Delta^\circ)}\, . 
\end{equation}
$\# \tilde E$ and $\# \tilde F $ are triangulation-independent and are entirely determined by $\Delta^\circ$~\cite{DeLoera:2010:TSA:1952022}.

Two $3d$ reflexive polytopes give a far larger number $|S_{\Delta^\circ}|$
than the others. They
are  the convex hulls $\Delta_i^\circ := \text{Conv}(S_i), i=1,2$ of the
vertex sets
\begin{align} 
S_1 &= \{ (-1,-1,-1),(-1,-1,5),(-1,5,-1),(1,-1,-1)\}\, , \nonumber \\
S_2 &= \{ (-1,-1,-1),(-1,-1,11),(-1,2,-1),(1,-1,-1)\}. \nonumber
\end{align}
$\mathcal{T}(\doc)$ and $\mathcal{T}(\dtc)$ have the same number of edges and faces. 
Their largest facets are displayed in  Fig. \ref{fig:bigfacetbigone1} and have $\# \tilde E = 63$ and $\# \tilde F=36$. We compute
\begin{equation}
|\sdoc| = \frac{2.96}{3} \times 10^{755} \qquad |\sdtc| = 2.96 \times 10^{755},
\label{eqn:sdocsdtccounts}
\end{equation}
where the factor of $1/3$ is due to the symmetries discussed in the Appendix.
All other polytopes $\Delta^\circ$ contribute negligibly:
$|S_{\Delta^\circ}| \leq 3.28\times 10^{692}$ 
configurations. This gives
\begin{equation}
\text{\# 4d F-theory Geometries} \geq \frac43 \times 2.96 \times 10^{755},
\end{equation}
which undercounts due to the facts that we choose to do face blowups followed by
edge blowups to simplify the subdivision combinatorics, and
that we have not taken into account the $O(10^{15})$ FRSTs of 
$\dtc$ and $\doc$.

\vspace{.2cm}
\noindent{\bf IV. Universality and Non-Higgsable Clusters.}

We now study universality in the dominant sets of
F-theory geometries $\sdoc$ and $\sdtc$. We prove non-Higgsable cluster universality, minimal gauge group universality,
and discuss results from random sampling.
\begin{figure}[t]
\begin{tikzpicture}[scale=.8]
\draw[thick,color=Black] (0,0) -- (3,0) -- (0,3) -- cycle;
\draw[thick,color=Black] (0,.5) -- (2.5,.5);
\draw[thick,color=Black] (0,1) -- (2,1);
\draw[thick,color=Black] (0,1.5) -- (1.5,1.5);
\draw[thick,color=Black] (0,2) -- (1,2);
\draw[thick,color=Black] (0,2.5) -- (.5,2.5);
\draw[thick,color=Black] (.5,0) -- (.5,2.5);
\draw[thick,color=Black] (1,0) -- (1,2);
\draw[thick,color=Black] (1.5,0) -- (1.5,1.5);
\draw[thick,color=Black] (2,0) -- (2,1);
\draw[thick,color=Black] (2.5,0) -- (2.5,.5);
\draw[thick,color=Black] (0,2) -- (.5,2.5);
\draw[thick,color=Black] (0,1) -- (1,2);
\draw[thick,color=Black] (0,0) -- (1.5,1.5);
\draw[thick,color=Black] (1,0) -- (2,1);
\draw[thick,color=Black] (2,0) -- (2.5,.5);
\draw[thick,color=Black] (0,1.5) -- (.5,2);
\draw[thick,color=Black] (0,.5) -- (1,1.5);
\draw[thick,color=Black] (.5,0) -- (1.5,1);
\draw[thick,color=Black] (1.5,0) -- (2,.5);
\fill (0,0) circle (.5mm); \fill (0,.5) circle (.5mm); \fill (0,1) circle (.5mm);
\fill (0,1.5) circle (.5mm); \fill (0,2) circle (.5mm); \fill (0,2.5) circle (.5mm);
\fill (0,3) circle (.5mm);
\fill (.5,0) circle (.5mm); \fill (.5,.5) circle (.5mm); \fill (.5,1) circle (.5mm);
\fill (.5,1.5) circle (.5mm); \fill (.5,2) circle (.5mm); \fill (.5,2.5) circle (.5mm);
\fill (1,0) circle (.5mm); \fill (1,.5) circle (.5mm); \fill (1,1) circle (.5mm);
\fill (1,1.5) circle (.5mm); \fill (1,2) circle (.5mm); 
\fill (1.5,0) circle (.5mm); \fill (1.5,.5) circle (.5mm); \fill (1.5,1) circle (.5mm);
\fill (1.5,1.5) circle (.5mm); 
\fill (2,0) circle (.5mm); \fill (2,.5) circle (.5mm); \fill (2,1) circle (.5mm);
\fill (2.5,0) circle (.5mm); \fill (2.5,.5) circle (.5mm);
\fill (3,0) circle (.5mm);
\draw[thick,color=Black] (1,3) -- (7,3) -- (7,0) -- cycle;
\draw[thick,color=Black] (3,2) -- (7,2);
\draw[thick,color=Black] (5,1) -- (7,1);
\draw[thick,color=Black] (6.5,1) -- (6.5,3);
\draw[thick,color=Black] (6,1) -- (6,3);
\draw[thick,color=Black] (5.5,1) -- (5.5,3);
\draw[thick,color=Black] (5,1) -- (5,3);
\draw[thick,color=Black] (4.5,2) -- (4.5,3);
\draw[thick,color=Black] (4,2) -- (4,3);
\draw[thick,color=Black] (3.5,2) -- (3.5,3);
\draw[thick,color=Black] (3,2) -- (3,3);
\draw[thick,color=Black] (7,1) -- (6,3);
\draw[thick,color=Black] (6.5,1) -- (5.5,3);
\draw[thick,color=Black] (6,1) -- (5,3);
\draw[thick,color=Black] (5.5,1) -- (4.5,3);
\draw[thick,color=Black] (5,1) -- (4,3);
\draw[thick,color=Black] (7,2) -- (6.5,3);
\draw[thick,color=Black] (4,2) -- (3.5,3);
\draw[thick,color=Black] (3.5,2) -- (3,3);
\draw[thick,color=Black] (7,0) -- (6.5,1);
\draw[thick,color=Black] (7,0) -- (6,1);
\draw[thick,color=Black] (7,0) -- (5.5,1);
\draw[thick,color=Black] (5,1) -- (4.5,2);
\draw[thick,color=Black] (5,1) -- (4,2);
\draw[thick,color=Black] (5,1) -- (3.5,2);
\draw[thick,color=Black] (3,2) -- (2.5,3);
\draw[thick,color=Black] (3,2) -- (2,3);
\draw[thick,color=Black] (3,2) -- (1.5,3);
\fill (1,3) circle (.5mm); \fill (1.5,3) circle (.5mm); \fill (2,3) circle (.5mm); \fill (2.5,3) circle (.5mm);
\fill (3,3) circle (.5mm); \fill (3.5,3) circle (.5mm); \fill (4,3) circle (.5mm); \fill (4.5,3) circle (.5mm);
\fill (5,3) circle (.5mm); \fill (5.5,3) circle (.5mm); \fill (6,3) circle (.5mm); \fill (6.5,3) circle (.5mm);
\fill (7,3) circle (.5mm);
\fill (7,2) circle (.5mm);
\fill (7,1) circle (.5mm);
\fill (7,0) circle (.5mm);
\fill (6.5,1) circle (.5mm);
\fill (6,1) circle (.5mm);
\fill (5.5,1) circle (.5mm);
\fill (5,1) circle (.5mm);
\fill (6.5,2) circle (.5mm);
\fill (6,2) circle (.5mm);
\fill (5.5,2) circle (.5mm);
\fill (5,2) circle (.5mm);
\fill (4.5,2) circle (.5mm);
\fill (4,2) circle (.5mm);
\fill (3.5,2) circle (.5mm);
\fill (3,2) circle (.5mm);

\end{tikzpicture}

\caption{The largest facets in the two 3d reflexive polytopes $\Delta_1^\circ$ and $\Delta_2^\circ$
with the most number of interior points. Presented is one triangulation of each,
from which we see $\#\tilde E=63$ edges and $\#\tilde F=36$ faces in both facets.}
\label{fig:bigfacetbigone1}
\end{figure}
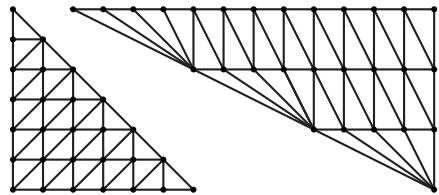

\vspace{.2cm}
\noindent \textbf{Algorithmic Universality and Gauge Groups.} We wish to establish
the likelihood that an F-theory base in $\sdoc$ or $\sdtc$ give rise to
non-Higgsable seven-branes. The result arises from Prop. \ref{prop:NH7fromsingletree}:
if there is a tree anywhere on $F$, even a single leaf, there is a
non-Higgsable seven-brane on all divisors associated to interior points of $F$.
For any $S_{\Delta^\circ}$ only one configuration has no trees, and therefore
\begin{equation}
P(\text{NHC in } S_{\Delta^\circ}) \geq 1 - \frac{1}{|S_{\Delta^\circ}|}.
\end{equation}
This is always very close to one, and in particular
\begin{align}
P(\text{NHC in } \sdoc) &\geq1-1.01\times 10^{-755} \nonumber \\
P(\text{NHC in } \sdtc) &\geq 1-.338\times 10^{-755}.
\end{align}
We see that NHC are universal in these ensembles.

We now wish to study physics in
our ensemble. Consider a geometric assumption $A_i$
and a physical property $P_i$ such that $A_i\implies P_i$. Our goal is to determine high
probability assumptions that lead to interesting physical properties, computing $P(A_i)$
since $A_i\implies P_i$ ensures $P(P_i)\geq P(A_i).$ We will focus on $\sdoc$ and $\sdtc$ since
these are the dominant ensembles.

Consider first $\sdoc$ and let $A_1$ be the assumption that any simplex in an FRST 
of $\doc$ containing a vertex of $\doc$ has an $h\geq 3$ face tree on it. 
For the $3$ symmetric facets of $\doc$ there are $17$ ways to choose simplices
containing the vertices, and $1796$ ways for its largest facet. The maximum
number of simplices containing vertices is $24$. Using $P(h\geq 3 \,\, \text{tree on simplex})$ from Table \ref{tab:numedgefacetreeandprob},
\begin{equation}
P(A_1\textin\sdoc)\geq .9999998^{24}=.999995.
\end{equation}
There are $17^3\times1796$ ways to choose simplices that contain the vertices,
all of which yield $G\geq F_4^{18}\times E_6^{10}\times U^9$ where 
$U\in\{G_2,F_4,E_6\}$, depending on details. All of these factors arise on
the ground, and generally there will be many more factors from non-Higgsable seven-branes
in the leaves. Here $E_6^{10}$ arises from an $E_6$ on every interior point
of the large facet in $\doc$, see Fig. \ref{fig:bigfacetbigone1}. This set of
statements defines physical property $P_1$, and since $A_1\implies P_1$
we deduce $P(P_1\textin \sdoc)\geq P(A_1\textin \sdoc) \geq .999995$.

Let $A_2$ be the assumption that there exists a $h=5$ face tree somewhere
on the large facet $F$ in $\doc$. Knowing $\tilde F=36$ on $F$ and using
Table \ref{tab:numedgefacetreeandprob}, we compute $P(A_2 \textin \sdoc)=(1-(1-.999997)^{36})\simeq 1-10^{-199}.$ Let $A_3$ be that $A_1$ and $A_2$ hold, so 
$P(A_3)=P(A_1)P(A_2)\simeq P(A_1)$. Then given $A_3$ a short calculation shows 
that the  $h=5$ tree on $F$ 
enhances $E_6$ in $P_1$ to $E_8$, giving $10$ $E_8$'s on the ground. $P_1$ with this enhancement defines $P_3$. 

Similar results hold for $\sdtc$. Let $A_1$ be the assumption that any simplex in an FRST 
of $\dtc$ containing a vertex of $\dtc$ has an $h\geq 3$ face tree on it. This ensures that $G\geq F_4^{15}\times E_6^{7}\times U^{12}$. However, this is quickly enhanced to  $G\geq F_4^{18}\times E_8^{10}\times U^9$, via a $h=5$ face tree on each face, and a $1-6.55\times 10^{-8}$ probability blow-up along an edge connecting the point $\{-1, 2, -1\}$ to one of the points $\{-1, 1, n\}$, where $n = -1\dots 3$. The existence of these edges is independent of triangulation. Summarizing, the probability that a geometry in our set has
$G\geq F_4^{18}\times E_8^{10}\times U^9$ on the ground is $\geq .999995$. This minimal group
for $P_3$ on $\sdtc$ matches that of $P_3$ on $\sdoc$.

It is natural to ask whether this structure on the ground constrains the gauge
structure in the trees. In Prop.~\ref{prop:E8roots} it is shown that the 
 gauge group on a leaf $v$ in a tree built above $E_8$'s on the ground
is determined by the leaf height $h_v$. 
The result is that a $h_v=1,2,3,4,5,6$ leaf
above $E_8$ roots has Kodaira fiber $F_v=II^*,IV^*_{ns},I^*_{0,ns},IV_{ns},II,-$
with gauge group $G_v=E_8,F_4,G_2,SU(2),-,-$, respectively.

This leads to a high probability result about the structure of the  gauge group.
Since $A_3\implies P_3$, which has at least $10$ $E_8$ factors nearby one another, $P_3$
also has
\begin{equation}
G\geq E_8^{10} \times F_4^{18}\times U^9 \times F_4^{H_2}\times G_2^{H_3}\times A_1^{H_4},
\end{equation}
where $H_i$ is the number of height $i$ leaves in trees built on $E_8$ roots,
and $rk(G)\geq 160+4H_2+2H_3+H_4$.
There are $H_6$ Kodaira type $II$ seven-branes that do not carry a gauge group but realize
Argyres-Douglas theories on D3 probes. The first $F_4$ and also
the $U$ factors may enhance, but the other factors are fixed. The probability
of this physical property is $P(P_3)\geq P(A_3) \simeq .999995$. This non-trivial minimal gauge structure is universal in our large ensembles given by $\sdoc$ and $\sdtc$.

\vspace{.2cm}
\noindent \textbf{Random Samples and Geometric Visible Sectors.} 
It may be possible to accommodate visible sectors from flux breaking
these gauge sectors, but it is also interesting to study whether gauge
factors $E_6$ and/or $SU(3)$ arise with high probability.
We have not yet discovered a high probability simple geometric assumption
that leads to $E_6$ or $SU(3)$. However, it is possible that they arise regularly,
but due to a complex geometric assumption.

This idea can be tested by random sampling. Let $B$ be an F-theory base obtained
by adding face trees then edge trees at random, followed by edge trees at random, to the ``pushing" 
triangulation~\cite{DeLoera:2010:TSA:1952022} of $\doc$. We studied an ensemble $S_r$ of $10^6$ such random samples 
and found $P(SU(3) \, \text{or} \, E_6 \textin S_r)\simeq 1/200$, and that at least 
$36$ of the points in $\doc$ carried $E_8$, a significant enhancement beyond
$P_3$.  Furthermore, in our sample we found that $E_6$ only
arose on the point $(1,-1,-1)$, which is the only vertex of $\doc$
that is not in the largest facet.  Similar results and probabilities also hold
using these techniques on $\dtc$. It would be interesting to study random samples of other triangulations, or to
see if other geometric assumptions imply these enhancements. We 
leave the systematic study of geometric visible sectors to future work.

\vspace{.2cm}
\noindent \textbf{V. Discussion.} 
We have presented a construction algorithm for  $\frac43 \times  2.96\times
10^{755}$ geometries for $4d$ F-theory compactifications. This number is only
a lower bound and may be enlarged in at least three ways: by relaxing
the requirement of edge blowups after  face blowups, by taking into account
the $O(10^{15})$ FRS triangulations of 3d reflexive polytopes, and by
considering blow-ups of non-toric intersections of seven-branes.

We have initiated the study of this ensemble by focusing on the geometric gauge group.
Using knowledge of the construction algorithm, we derived the existence of universal
properties for the minimal geometric gauge group on non-Higgsable clusters. High
rank groups are generic, as is the existence of at least $10$ $E_8$ factors on the ground.
The gauge group on leaves above these $E_8$ factors on the ground is fixed entirely by their height. Such large gauge sectors motivate dark glueballs; see \cite{Halverson:2016nfq,*Soni:2016yes,*daRocha:2017cxu,*Acharya:2017szw,*Soni:2017nlm}.

There are many directions for future work. For example, it would be
interesting to study the number of consistent fluxes per geometry and how they alter the gauge group, to
perform a statistical analysis of the gauge group in the leaves, or to analyze
physics that arises from blow-ups of non-toric intersections of seven-branes. 
Perhaps most pressing is that, though we have demonstrated that the gauge group is generically 
high rank and have reviewed some possible realizations of the standard model discussed in \cite{Grassi:2014zxa}, it is not yet clear whether the standard model is realized with high probability in our ensemble.

We believe that this is the first time that such a large ensemble has been systematically
studied in string theory. In our view, the crucial ingredient that made the results
possible are what we call algorithmic
universality: derivation of universality from a construction algorithm, rather than
an explicitly constructed ensemble or random sampling. Given the plethora of large ensembles in string theory
and the infeasibility of constructing all of them, universality of this sort may play
a critical role in making the string landscape tractable.

\vspace{.2cm}
\noindent \textbf{A. Appendix: Technical Subtleties.} We now
address technical subtleties that are important for establishing,
but not understanding, results in the main text.

\vspace{.2cm}
\noindent \textbf{Polytope Symmetries and Toric Morphisms.}
In equation \eqref{eqn:sdocsdtccounts} 
we have included a factor of
$1/3$ relative to the count one would obtain directly from the algorithm.
This takes into account an overcounting of geometries due to toric equivalences,
which arise when there is a $GL(3,\bZ)$ transformation on the toric rays
that preserves the cone structure of the fan. In general, there may
be many such equivalences between elements of two ensembles $S_{\Delta_i^\circ}$
and $S_{\Delta_j^\circ}$, where $\Delta_{i,j}^\circ$ are any two $3d$ reflexive
polytopes. However, to ensure that the count  \eqref{eqn:sdocsdtccounts} is accurate,
we only need to consider whether there are equivalences between
two elements in $\sdoc$, two in $\sdtc$, or one in $\sdoc$ to $\sdtc$. It is sufficient to consider $GL(3,\bZ)$ actions
on the ground, i.e. on the facets. This follows from the fact that rays of different height cannot be exchanged under automorphisms of the fan. First note that points in a hyperplane remain in one after a $GL(3,\bZ)$ transformation, and points
in a line remain in a line by linearity. Facets must therefore map to
facets. The big facets in $\doc$ and $\dtc$ cannot map
to other facets by point counting, and therefore they
must map to themselves. There is no non-trivial map
taking the big facet in $\dtc$ to itself, but there is
a $\bZ_3$ rotation taking the big facet in $\doc$ to itself, giving a factor of $1/3$ in $\sdoc$. There is no non-trivial
map between the big facets in $\doc$ and $\dtc$, and therefore
$\sdtc\cap \sdoc  = \emptyset.$ Together, these establish
\eqref{eqn:sdocsdtccounts}.

\vspace{.2cm}
\noindent \textbf{Multiplicities of Vanishing and Resolutions.}
In discussing what constitutes an allowed $4d$ F-theory geometry $X\to B$, we mentioned
certain criteria on multiplicities of vanishing that we now elaborate on. In 
\cite{Hayakawa,Wang} it was shown that if a Calabi-Yau variety has at worst
canonical singularities, then it is at finite distance 
from the bulk of the moduli space in the Weil-Petersson metric. 
This criterion is general and therefore applies to elliptic fibrations such as $X$.
The reason that it is physically relevant is that if $X$ has worse singularities
than a nearby Calabi-Yau $X'$ that is known to represent a physical configuration,
and $X$ is at finite distance in the moduli space from $X'$, we should expect that
$X$ is also a physical configuration. This criterion, which we refer to as the
Hayakawa-Wang criterion, gives a related criterion by studying elliptic
fibrations~\cite{Grassi1991, Candelas:2000nc, Morrisonunp}:
if $mult_D (f,g) < (4,6)$, $mult_C(f,g) < (8,12)$, and $ord_p(f,g) < (12,18)$ 
for all divisors $D\subset B$, curves $C\subset B$, and points $p\subset B$, respectively,
then  $X$ has at worst
canonical singularities and is at finite distance in the
moduli space due to the Hayakawa-Wang criterion\footnote{We thank D. Morrison for discussions on
this and related points.}. Here ``less than'' means that at least one of the multiplicities or orders is strictly less than the given multiplicity or order. We now translate this multiplicity of vanishing (MOV) condition to constraints on the height of the tree.
\begin{prop}
\label{prop:heightcrit}
Suppose each leaf $v\in B$ has height $h_v\leq 6$. Then $B$ has
no $(4,6)$ divisors.
\end{prop}

\begin{proof}
Consider a facet
$F$, which has a unique associated point $m_F$ satisfying
$(m_F,\tilde v)=-1 \,\, \forall \tilde v \in F$; furthermore
since $m_F\in\Delta$, 
$(m_F,v)\geq -1\,\, \forall v\in\Delta^\circ$. Now
suppose $h_v \leq 6/n, \,\, n \in \mathbb{N}$  
for all rays $v=av_1 + b v_2 + c v_3$
in $B$, with $v_i$ $3d$ cone vertices in $B_i$ . Then 
$(nm_F, v)\geq -n(a+b+c)=-nh_v\geq -6$ for all rays $v$ and
therefore $nm_F \in \Delta_g$. Here we denote $\Delta_{f},\Delta_{g}$ as the polytopes corresponding to $\Gamma (\mathcal{O}(-4 K_B)),\Gamma (\mathcal{O}(-6 K_B))$, respectively. If $h_v\leq 6\,\, \forall v$, then $m_F\in\Delta_g$. This monomial has multiplicity of vanishing
$(v,m_F)+6=5$ for any $v$ in or above $F$, which protects
$v$ from being a $(4,6)$ divisor. If
$h_v\leq 6 \,\, \forall v$ then $m_F \in \Delta_g\,\, \forall F$
and there is a monomial that prevents
each divisor from being $(4,6)$.
\end{proof}

It is also simple to see that in our ensemble, $f$ and $g$ can only vanish to multiplicities less than $(8,12)$ along curves and orders less than $(12,18)$ at points, respectively. Consider any toric curve $C=D_s\cdot D_t \subset B$. Take $v_s=\sum_i a_{i,s} v_i$ and $v_t=\sum_i a_{i,t} v_i$ and define $a:=\sum_i a_{i,s}$ and $b:=\sum_i a_{i,t}$. Let $F$ be a facet
on which or above which $v_s$ and $v_t$ sit, with $m_F$ the dual facet. As an
element of $\Delta_g$ the associated monomial may be written
$s^{\vev{m,v_s}+6}t^{\vev{m,v_t}+6}\times \dots$,
and the monomial vanishes to multiplicity $\vev{m,v_s}+\vev{m,v_t}+12=-a-b+12$ along $C$. For $g$ to vanish to multiplicity $12$ along a curve, this requires $a+ b <0$, which cannot happen. A similar argument shows that $g$ cannot vanish to order $18$ or higher at points.
On the other hand, our ensemble is generated by a series of repeated blowups along curves and points, and one can pass to a Calabi-Yau minimal Weierstrass model only if the MOV is $\geq (4,6)$ for a curve, and $\geq (8,12)$  for a point. One can achieve the required MOV by tuning in complex structure moduli space, but one has to ensure that no infinite distance singularities (as in the above) are introduced in the process. However, it is simple to see that the desired MOV can be achieved, without introducing any disallowed singularities, by simple tuning without turning off the monomial corresponding to $m_F$, for all $F$.

\vspace{.2cm}
\noindent \textbf{7-Branes and Gauge Enhancement.}
We now prove some useful results that allow us to determine a universal minimal gauge sector in our ensemble, as well as show that NH 7-branes are ubiquitous.
\begin{prop}
\label{prop:NH7fromsingletree}
Suppose $\exists$ $v$ in or above a facet $F$,i.e. $v=av_1+bv_2+cv_3$ with $v_i$ simplex vertices in $F$, such that $h_v\geq 2$. Then there is a non-Higgsable seven-brane on the divisor associated
to each interior point of $F$.
\end{prop}

\begin{proof}
Then $(6 m_F,v)=-6h_v \leq -12$ implies $6m_F \notin \Delta_g$.
Similarly, $4m_F \notin \Delta_f$. Since any point $p$ interior to $F$ has 
$(m,p)=-1 \iff m=m_F$ and reflexive polytopes of dimension three
are normal, i.e. any $m_f\in \Delta_f$ ($m_f \in \Delta_g$) has
$m_f = \sum_i m_i, m_i\in \Delta$ ($m_g = \sum_i m_i, m_i\in \Delta$),
it follows that $(m_f,p)=-4 \iff m_f = 4m_F$ and $(m_g,p)=-6 \iff m_g = 6m_F$. Therefore, if there is any tree on $F$
then $4m_F \notin \Delta_f$ and $6m_F\notin \Delta_g$. By normality,
for any $p$ interior to $F$ this gives
$\nexists m_f \in \Delta_f | (m_f,p)=-4$ and
$\nexists m_g \in \Delta_g | (m_g,p)=-6$, and therefore $ord_p(f,g) > (0,0)$,
which implies there is a non-Higgsable seven-brane on the divisor associated to $p$.
\end{proof}

\begin{prop}
\label{prop:E8roots}
Let $v$ be a leaf $v=av_1 + bv_2 + cv_3$ with $v_i$ simplex vertices in $F$. If the associated divisors $D_{1,2,3}$
carry a non-Higgsable $E_8$ seven-brane, and if $v$ has height
 $h_v=1,2,3,4,5,6$ it also has Kodaira fiber $F_v=II^*,IV^*_{ns},I^*_{0,ns},IV_{ns},II,-$
and gauge group $G_v=E_8,F_4,G_2,SU(2),-,-$, respectively.
\end{prop}
\begin{proof}
The height criterion
gives $mult_v(g)\leq 6-h_v$. If $v=av_1+bv_2+cv_3$ with $v_i$ each carrying $E_8$,
then 
$(m_f,v_i)\geq 0, (m_g,v_i)\geq -1\,\,, \forall m_f\in \Delta_f 
$ and $\forall m_g\in \Delta_g$.
This gives $(m_f,v)\geq 0$, $(m_g,v)\geq-(a+b+c)= -h_v$. Together, we see
$mult_v(f)\geq 4$, $mult_v(g)=6-h_v$. For $h_v= 1,5,6$ this fixes $G_v$, but 
to determine $G_v$ for $h_v=2,3,4$ we must study the split condition. A necessary
condition is that there is one monomial $m_g\in \Delta_g$ such that $(m_g,v)+6=6-h_v$,
and since $m_F \in \Delta_g$ always, where $F$ is the facet in which $v_i$ lie,
then $m_g=m_F$. Moreover, the monomial $m$ in $g$ associated to $m_F$ must be a perfect
square; since $(m_F,v_i)+6=5$, $m\sim x_i^5$ and $m$ is not a perfect square. 
Therefore the fibers are all non-split. This establishes the result.
\end{proof}

\noindent{\bf Blowdowns and Oda's Factorization Conjectures.}  We have obtained all trees from a sequence of blow-ups from an
initial triangle on the ground. It may be possible to arrive at additional consistent tree
configurations via  blowing down at intermediate steps. We did not consider
such possibilities, for combinatorial reasons.

However, such questions about mixing blow-ups and blow-downs are the subject
of Oda's Weak and Strong Factorization conjectures. The former states that any proper birational morphism $X \dashedrightarrow Y$ of complete, nonsingular varieties in characteristic zero factors into a sequence of smooth blow-ups and blow-downs. The latter conjectures
that the morphism factors into a sequence of successive blow-ups followed by a sequence of
successive blow-downs; it is open in dimension 3 and higher.

An interesting physical question arises in this context. By weak Oda, two trees are
related by a sequence of blow-ups and blow-downs. However, if each sequence between
fixed $X$ and $Y$ gives rise to an intermediate variety $X_i$ with a $(4,6)$ divisor,
then the moduli space of four-dimensional F-theory compactifications is disconnected.

\vspace{.2cm}
\noindent{\bf Acknowledgements.} We thank W. Cunningham, T. Eliassi-Rad, J. Goodrich,  D. Krioukov, B. Nelson, W. Taylor, J. Tian,
and especially D.R. Morrison for discussions. 
J.H. is supported by
NSF Grant PHY-1620526.

\bibliography{refs}

\begin{thebibliography}{35}%
\makeatletter
\providecommand \@ifxundefined [1]{%
 \@ifx{#1\undefined}
}%
\providecommand \@ifnum [1]{%
 \ifnum #1\expandafter \@firstoftwo
 \else \expandafter \@secondoftwo
 \fi
}%
\providecommand \@ifx [1]{%
 \ifx #1\expandafter \@firstoftwo
 \else \expandafter \@secondoftwo
 \fi
}%
\providecommand \natexlab [1]{#1}%
\providecommand \enquote  [1]{``#1''}%
\providecommand \bibnamefont  [1]{#1}%
\providecommand \bibfnamefont [1]{#1}%
\providecommand \citenamefont [1]{#1}%
\providecommand \href@noop [0]{\@secondoftwo}%
\providecommand \href [0]{\begingroup \@sanitize@url \@href}%
\providecommand \@href[1]{\@@startlink{#1}\@@href}%
\providecommand \@@href[1]{\endgroup#1\@@endlink}%
\providecommand \@sanitize@url [0]{\catcode `\\12\catcode `\$12\catcode
  `\&12\catcode `\#12\catcode `\^12\catcode `\_12\catcode `\%12\relax}%
\providecommand \@@startlink[1]{}%
\providecommand \@@endlink[0]{}%
\providecommand \url  [0]{\begingroup\@sanitize@url \@url }%
\providecommand \@url [1]{\endgroup\@href {#1}{\urlprefix }}%
\providecommand \urlprefix  [0]{URL }%
\providecommand \Eprint [0]{\href }%
\providecommand \doibase [0]{http://dx.doi.org/}%
\providecommand \selectlanguage [0]{\@gobble}%
\providecommand \bibinfo  [0]{\@secondoftwo}%
\providecommand \bibfield  [0]{\@secondoftwo}%
\providecommand \translation [1]{[#1]}%
\providecommand \BibitemOpen [0]{}%
\providecommand \bibitemStop [0]{}%
\providecommand \bibitemNoStop [0]{.\EOS\space}%
\providecommand \EOS [0]{\spacefactor3000\relax}%
\providecommand \BibitemShut  [1]{\csname bibitem#1\endcsname}%
\let\auto@bib@innerbib\@empty
\bibitem [{\citenamefont {Bousso}\ and\ \citenamefont
  {Polchinski}(2000)}]{Bousso:2000xa}%
  \BibitemOpen
  \bibfield  {author} {\bibinfo {author} {\bibfnamefont {R.}~\bibnamefont
  {Bousso}}\ and\ \bibinfo {author} {\bibfnamefont {J.}~\bibnamefont
  {Polchinski}},\ }\href {\doibase 10.1088/1126-6708/2000/06/006} {\bibfield
  {journal} {\bibinfo  {journal} {JHEP}\ }\textbf {\bibinfo {volume} {06}},\
  \bibinfo {pages} {006} (\bibinfo {year} {2000})},\ \Eprint
  {http://arxiv.org/abs/hep-th/0004134} {arXiv:hep-th/0004134 [hep-th]}
  \BibitemShut {NoStop}%
\bibitem [{\citenamefont {Ashok}\ and\ \citenamefont
  {Douglas}(2004)}]{Ashok:2003gk}%
  \BibitemOpen
  \bibfield  {author} {\bibinfo {author} {\bibfnamefont {S.}~\bibnamefont
  {Ashok}}\ and\ \bibinfo {author} {\bibfnamefont {M.~R.}\ \bibnamefont
  {Douglas}},\ }\href {\doibase 10.1088/1126-6708/2004/01/060} {\bibfield
  {journal} {\bibinfo  {journal} {JHEP}\ }\textbf {\bibinfo {volume} {01}},\
  \bibinfo {pages} {060} (\bibinfo {year} {2004})},\ \Eprint
  {http://arxiv.org/abs/hep-th/0307049} {arXiv:hep-th/0307049 [hep-th]}
  \BibitemShut {NoStop}%
\bibitem [{\citenamefont {Denef}\ and\ \citenamefont
  {Douglas}(2004)}]{Denef:2004ze}%
  \BibitemOpen
  \bibfield  {author} {\bibinfo {author} {\bibfnamefont {F.}~\bibnamefont
  {Denef}}\ and\ \bibinfo {author} {\bibfnamefont {M.~R.}\ \bibnamefont
  {Douglas}},\ }\href {\doibase 10.1088/1126-6708/2004/05/072} {\bibfield
  {journal} {\bibinfo  {journal} {JHEP}\ }\textbf {\bibinfo {volume} {05}},\
  \bibinfo {pages} {072} (\bibinfo {year} {2004})},\ \Eprint
  {http://arxiv.org/abs/hep-th/0404116} {arXiv:hep-th/0404116 [hep-th]}
  \BibitemShut {NoStop}%
\bibitem [{\citenamefont {Denef}\ and\ \citenamefont
  {Douglas}(2007)}]{Denef:2006ad}%
  \BibitemOpen
  \bibfield  {author} {\bibinfo {author} {\bibfnamefont {F.}~\bibnamefont
  {Denef}}\ and\ \bibinfo {author} {\bibfnamefont {M.~R.}\ \bibnamefont
  {Douglas}},\ }\href {\doibase 10.1016/j.aop.2006.07.013} {\bibfield
  {journal} {\bibinfo  {journal} {Annals Phys.}\ }\textbf {\bibinfo {volume}
  {322}},\ \bibinfo {pages} {1096} (\bibinfo {year} {2007})},\ \Eprint
  {http://arxiv.org/abs/hep-th/0602072} {arXiv:hep-th/0602072 [hep-th]}
  \BibitemShut {NoStop}%
\bibitem [{\citenamefont {Cvetic}\ \emph {et~al.}(2011)\citenamefont {Cvetic},
  \citenamefont {Garcia-Etxebarria},\ and\ \citenamefont
  {Halverson}}]{Cvetic:2010ky}%
  \BibitemOpen
  \bibfield  {author} {\bibinfo {author} {\bibfnamefont {M.}~\bibnamefont
  {Cvetic}}, \bibinfo {author} {\bibfnamefont {I.}~\bibnamefont
  {Garcia-Etxebarria}}, \ and\ \bibinfo {author} {\bibfnamefont
  {J.}~\bibnamefont {Halverson}},\ }\href {\doibase 10.1002/prop.201000093}
  {\bibfield  {journal} {\bibinfo  {journal} {Fortsch. Phys.}\ }\textbf
  {\bibinfo {volume} {59}},\ \bibinfo {pages} {243} (\bibinfo {year} {2011})},\
  \Eprint {http://arxiv.org/abs/1009.5386} {arXiv:1009.5386 [hep-th]}
  \BibitemShut {NoStop}%
\bibitem [{\citenamefont {Vafa}(1996)}]{Vafa:1996xn}%
  \BibitemOpen
  \bibfield  {author} {\bibinfo {author} {\bibfnamefont {C.}~\bibnamefont
  {Vafa}},\ }\href {\doibase 10.1016/0550-3213(96)00172-1} {\bibfield
  {journal} {\bibinfo  {journal} {Nucl. Phys.}\ }\textbf {\bibinfo {volume}
  {B469}},\ \bibinfo {pages} {403} (\bibinfo {year} {1996})},\ \Eprint
  {http://arxiv.org/abs/hep-th/9602022} {arXiv:hep-th/9602022 [hep-th]}
  \BibitemShut {NoStop}%
\bibitem [{\citenamefont {Morrison}\ and\ \citenamefont
  {Vafa}(1996)}]{Morrison:1996pp}%
  \BibitemOpen
  \bibfield  {author} {\bibinfo {author} {\bibfnamefont {D.~R.}\ \bibnamefont
  {Morrison}}\ and\ \bibinfo {author} {\bibfnamefont {C.}~\bibnamefont
  {Vafa}},\ }\href {\doibase 10.1016/0550-3213(96)00369-0} {\bibfield
  {journal} {\bibinfo  {journal} {Nucl. Phys.}\ }\textbf {\bibinfo {volume}
  {B476}},\ \bibinfo {pages} {437} (\bibinfo {year} {1996})},\ \Eprint
  {http://arxiv.org/abs/hep-th/9603161} {arXiv:hep-th/9603161 [hep-th]}
  \BibitemShut {NoStop}%
\bibitem [{\citenamefont {Morrison}\ and\ \citenamefont
  {Taylor}(2012{\natexlab{a}})}]{Morrison:2012np}%
  \BibitemOpen
  \bibfield  {author} {\bibinfo {author} {\bibfnamefont {D.~R.}\ \bibnamefont
  {Morrison}}\ and\ \bibinfo {author} {\bibfnamefont {W.}~\bibnamefont
  {Taylor}},\ }\href {\doibase 10.2478/s11534-012-0065-4} {\bibfield  {journal}
  {\bibinfo  {journal} {Central Eur. J. Phys.}\ }\textbf {\bibinfo {volume}
  {10}},\ \bibinfo {pages} {1072} (\bibinfo {year} {2012}{\natexlab{a}})},\
  \Eprint {http://arxiv.org/abs/1201.1943} {arXiv:1201.1943 [hep-th]}
  \BibitemShut {NoStop}%
\bibitem [{\citenamefont {Grassi}\ \emph {et~al.}(2015)\citenamefont {Grassi},
  \citenamefont {Halverson}, \citenamefont {Shaneson},\ and\ \citenamefont
  {Taylor}}]{Grassi:2014zxa}%
  \BibitemOpen
  \bibfield  {author} {\bibinfo {author} {\bibfnamefont {A.}~\bibnamefont
  {Grassi}}, \bibinfo {author} {\bibfnamefont {J.}~\bibnamefont {Halverson}},
  \bibinfo {author} {\bibfnamefont {J.}~\bibnamefont {Shaneson}}, \ and\
  \bibinfo {author} {\bibfnamefont {W.}~\bibnamefont {Taylor}},\ }\href
  {\doibase 10.1007/JHEP01(2015)086} {\bibfield  {journal} {\bibinfo  {journal}
  {JHEP}\ }\textbf {\bibinfo {volume} {01}},\ \bibinfo {pages} {086} (\bibinfo
  {year} {2015})},\ \Eprint {http://arxiv.org/abs/1409.8295} {arXiv:1409.8295
  [hep-th]} \BibitemShut {NoStop}%
\bibitem [{\citenamefont {Braun}\ and\ \citenamefont
  {Watari}(2015)}]{Braun:2014xka}%
  \BibitemOpen
  \bibfield  {author} {\bibinfo {author} {\bibfnamefont {A.~P.}\ \bibnamefont
  {Braun}}\ and\ \bibinfo {author} {\bibfnamefont {T.}~\bibnamefont {Watari}},\
  }\href {\doibase 10.1007/JHEP01(2015)047} {\bibfield  {journal} {\bibinfo
  {journal} {JHEP}\ }\textbf {\bibinfo {volume} {01}},\ \bibinfo {pages} {047}
  (\bibinfo {year} {2015})},\ \Eprint {http://arxiv.org/abs/1408.6167}
  {arXiv:1408.6167 [hep-th]} \BibitemShut {NoStop}%
\bibitem [{\citenamefont {Watari}(2015)}]{Watari:2015ysa}%
  \BibitemOpen
  \bibfield  {author} {\bibinfo {author} {\bibfnamefont {T.}~\bibnamefont
  {Watari}},\ }\href {\doibase 10.1007/JHEP11(2015)065} {\bibfield  {journal}
  {\bibinfo  {journal} {JHEP}\ }\textbf {\bibinfo {volume} {11}},\ \bibinfo
  {pages} {065} (\bibinfo {year} {2015})},\ \Eprint
  {http://arxiv.org/abs/1506.08433} {arXiv:1506.08433 [hep-th]} \BibitemShut
  {NoStop}%
\bibitem [{\citenamefont {Halverson}\ and\ \citenamefont
  {Tian}(2017)}]{Halverson:2016tve}%
  \BibitemOpen
  \bibfield  {author} {\bibinfo {author} {\bibfnamefont {J.}~\bibnamefont
  {Halverson}}\ and\ \bibinfo {author} {\bibfnamefont {J.}~\bibnamefont
  {Tian}},\ }\href {\doibase 10.1103/PhysRevD.95.026005} {\bibfield  {journal}
  {\bibinfo  {journal} {Phys. Rev.}\ }\textbf {\bibinfo {volume} {D95}},\
  \bibinfo {pages} {026005} (\bibinfo {year} {2017})},\ \Eprint
  {http://arxiv.org/abs/1610.08864} {arXiv:1610.08864 [hep-th]} \BibitemShut
  {NoStop}%
\bibitem [{\citenamefont {Halverson}(2017)}]{Halverson:2016vwx}%
  \BibitemOpen
  \bibfield  {author} {\bibinfo {author} {\bibfnamefont {J.}~\bibnamefont
  {Halverson}},\ }\href {\doibase 10.1016/j.nuclphysb.2017.02.014} {\bibfield
  {journal} {\bibinfo  {journal} {Nucl. Phys.}\ }\textbf {\bibinfo {volume}
  {B919}},\ \bibinfo {pages} {267} (\bibinfo {year} {2017})},\ \Eprint
  {http://arxiv.org/abs/1603.01639} {arXiv:1603.01639 [hep-th]} \BibitemShut
  {NoStop}%
\bibitem [{\citenamefont {Morrison}\ and\ \citenamefont
  {Taylor}(2015)}]{Morrison:2014lca}%
  \BibitemOpen
  \bibfield  {author} {\bibinfo {author} {\bibfnamefont {D.~R.}\ \bibnamefont
  {Morrison}}\ and\ \bibinfo {author} {\bibfnamefont {W.}~\bibnamefont
  {Taylor}},\ }\href {\doibase 10.1007/JHEP05(2015)080} {\bibfield  {journal}
  {\bibinfo  {journal} {JHEP}\ }\textbf {\bibinfo {volume} {05}},\ \bibinfo
  {pages} {080} (\bibinfo {year} {2015})},\ \Eprint
  {http://arxiv.org/abs/1412.6112} {arXiv:1412.6112 [hep-th]} \BibitemShut
  {NoStop}%
\bibitem [{\citenamefont {Taylor}\ and\ \citenamefont
  {Wang}(2015{\natexlab{a}})}]{Taylor:2015xtz}%
  \BibitemOpen
  \bibfield  {author} {\bibinfo {author} {\bibfnamefont {W.}~\bibnamefont
  {Taylor}}\ and\ \bibinfo {author} {\bibfnamefont {Y.-N.}\ \bibnamefont
  {Wang}},\ }\href {\doibase 10.1007/JHEP12(2015)164} {\bibfield  {journal}
  {\bibinfo  {journal} {JHEP}\ }\textbf {\bibinfo {volume} {12}},\ \bibinfo
  {pages} {164} (\bibinfo {year} {2015}{\natexlab{a}})},\ \Eprint
  {http://arxiv.org/abs/1511.03209} {arXiv:1511.03209 [hep-th]} \BibitemShut
  {NoStop}%
\bibitem [{\citenamefont {Halverson}\ and\ \citenamefont
  {Taylor}(2015)}]{Halverson:2015jua}%
  \BibitemOpen
  \bibfield  {author} {\bibinfo {author} {\bibfnamefont {J.}~\bibnamefont
  {Halverson}}\ and\ \bibinfo {author} {\bibfnamefont {W.}~\bibnamefont
  {Taylor}},\ }\href {\doibase 10.1007/JHEP09(2015)086} {\bibfield  {journal}
  {\bibinfo  {journal} {JHEP}\ }\textbf {\bibinfo {volume} {09}},\ \bibinfo
  {pages} {086} (\bibinfo {year} {2015})},\ \Eprint
  {http://arxiv.org/abs/1506.03204} {arXiv:1506.03204 [hep-th]} \BibitemShut
  {NoStop}%
\bibitem [{\citenamefont {Taylor}\ and\ \citenamefont
  {Wang}(2016)}]{Taylor:2015ppa}%
  \BibitemOpen
  \bibfield  {author} {\bibinfo {author} {\bibfnamefont {W.}~\bibnamefont
  {Taylor}}\ and\ \bibinfo {author} {\bibfnamefont {Y.-N.}\ \bibnamefont
  {Wang}},\ }\href {\doibase 10.1007/JHEP01(2016)137} {\bibfield  {journal}
  {\bibinfo  {journal} {JHEP}\ }\textbf {\bibinfo {volume} {01}},\ \bibinfo
  {pages} {137} (\bibinfo {year} {2016})},\ \Eprint
  {http://arxiv.org/abs/1510.04978} {arXiv:1510.04978 [hep-th]} \BibitemShut
  {NoStop}%
\bibitem [{\citenamefont {Morrison}\ and\ \citenamefont
  {Taylor}(2012{\natexlab{b}})}]{Morrison:2012js}%
  \BibitemOpen
  \bibfield  {author} {\bibinfo {author} {\bibfnamefont {D.~R.}\ \bibnamefont
  {Morrison}}\ and\ \bibinfo {author} {\bibfnamefont {W.}~\bibnamefont
  {Taylor}},\ }\href {\doibase 10.1002/prop.201200086} {\bibfield  {journal}
  {\bibinfo  {journal} {Fortsch. Phys.}\ }\textbf {\bibinfo {volume} {60}},\
  \bibinfo {pages} {1187} (\bibinfo {year} {2012}{\natexlab{b}})},\ \Eprint
  {http://arxiv.org/abs/1204.0283} {arXiv:1204.0283 [hep-th]} \BibitemShut
  {NoStop}%
\bibitem [{\citenamefont {Taylor}(2012)}]{Taylor:2012dr}%
  \BibitemOpen
  \bibfield  {author} {\bibinfo {author} {\bibfnamefont {W.}~\bibnamefont
  {Taylor}},\ }\href {\doibase 10.1007/JHEP08(2012)032} {\bibfield  {journal}
  {\bibinfo  {journal} {JHEP}\ }\textbf {\bibinfo {volume} {08}},\ \bibinfo
  {pages} {032} (\bibinfo {year} {2012})},\ \Eprint
  {http://arxiv.org/abs/1205.0952} {arXiv:1205.0952 [hep-th]} \BibitemShut
  {NoStop}%
\bibitem [{\citenamefont {Morrison}\ and\ \citenamefont
  {Taylor}(2014)}]{Morrison:2014era}%
  \BibitemOpen
  \bibfield  {author} {\bibinfo {author} {\bibfnamefont {D.~R.}\ \bibnamefont
  {Morrison}}\ and\ \bibinfo {author} {\bibfnamefont {W.}~\bibnamefont
  {Taylor}},\ }\href@noop {} {\  (\bibinfo {year} {2014})},\ \Eprint
  {http://arxiv.org/abs/1404.1527} {arXiv:1404.1527 [hep-th]} \BibitemShut
  {NoStop}%
\bibitem [{\citenamefont {Martini}\ and\ \citenamefont
  {Taylor}(2015)}]{Martini:2014iza}%
  \BibitemOpen
  \bibfield  {author} {\bibinfo {author} {\bibfnamefont {G.}~\bibnamefont
  {Martini}}\ and\ \bibinfo {author} {\bibfnamefont {W.}~\bibnamefont
  {Taylor}},\ }\href {\doibase 10.1007/JHEP06(2015)061} {\bibfield  {journal}
  {\bibinfo  {journal} {JHEP}\ }\textbf {\bibinfo {volume} {06}},\ \bibinfo
  {pages} {061} (\bibinfo {year} {2015})},\ \Eprint
  {http://arxiv.org/abs/1404.6300} {arXiv:1404.6300 [hep-th]} \BibitemShut
  {NoStop}%
\bibitem [{\citenamefont {Johnson}\ and\ \citenamefont
  {Taylor}(2014)}]{Johnson:2014xpa}%
  \BibitemOpen
  \bibfield  {author} {\bibinfo {author} {\bibfnamefont {S.~B.}\ \bibnamefont
  {Johnson}}\ and\ \bibinfo {author} {\bibfnamefont {W.}~\bibnamefont
  {Taylor}},\ }\href {\doibase 10.1007/JHEP10(2014)023} {\bibfield  {journal}
  {\bibinfo  {journal} {JHEP}\ }\textbf {\bibinfo {volume} {10}},\ \bibinfo
  {pages} {23} (\bibinfo {year} {2014})},\ \Eprint
  {http://arxiv.org/abs/1406.0514} {arXiv:1406.0514 [hep-th]} \BibitemShut
  {NoStop}%
\bibitem [{\citenamefont {Taylor}\ and\ \citenamefont
  {Wang}(2015{\natexlab{b}})}]{Taylor:2015isa}%
  \BibitemOpen
  \bibfield  {author} {\bibinfo {author} {\bibfnamefont {W.}~\bibnamefont
  {Taylor}}\ and\ \bibinfo {author} {\bibfnamefont {Y.-N.}\ \bibnamefont
  {Wang}},\ }\href@noop {} {\  (\bibinfo {year} {2015}{\natexlab{b}})},\
  \Eprint {http://arxiv.org/abs/1504.07689} {arXiv:1504.07689 [hep-th]}
  \BibitemShut {NoStop}%
\bibitem [{\citenamefont {Kreuzer}\ and\ \citenamefont
  {Skarke}(1998)}]{Kreuzer:1998vb}%
  \BibitemOpen
  \bibfield  {author} {\bibinfo {author} {\bibfnamefont {M.}~\bibnamefont
  {Kreuzer}}\ and\ \bibinfo {author} {\bibfnamefont {H.}~\bibnamefont
  {Skarke}},\ }\href@noop {} {\bibfield  {journal} {\bibinfo  {journal} {Adv.
  Theor. Math. Phys.}\ }\textbf {\bibinfo {volume} {2}},\ \bibinfo {pages}
  {847} (\bibinfo {year} {1998})},\ \Eprint
  {http://arxiv.org/abs/hep-th/9805190} {arXiv:hep-th/9805190 [hep-th]}
  \BibitemShut {NoStop}%
\bibitem [{\citenamefont {De~Loera}\ \emph {et~al.}(2010)\citenamefont
  {De~Loera}, \citenamefont {Rambau},\ and\ \citenamefont
  {Santos}}]{DeLoera:2010:TSA:1952022}%
  \BibitemOpen
  \bibfield  {author} {\bibinfo {author} {\bibfnamefont {J.~A.}\ \bibnamefont
  {De~Loera}}, \bibinfo {author} {\bibfnamefont {J.}~\bibnamefont {Rambau}}, \
  and\ \bibinfo {author} {\bibfnamefont {F.}~\bibnamefont {Santos}},\
  }\href@noop {} {\emph {\bibinfo {title} {Triangulations: Structures for
  Algorithms and Applications}}},\ \bibinfo {edition} {1st}\ ed.\ (\bibinfo
  {publisher} {Springer Publishing Company, Incorporated},\ \bibinfo {year}
  {2010})\BibitemShut {NoStop}%
\bibitem [{\citenamefont {Halverson}\ \emph {et~al.}(2017)\citenamefont
  {Halverson}, \citenamefont {Nelson},\ and\ \citenamefont
  {Ruehle}}]{Halverson:2016nfq}%
  \BibitemOpen
  \bibfield  {author} {\bibinfo {author} {\bibfnamefont {J.}~\bibnamefont
  {Halverson}}, \bibinfo {author} {\bibfnamefont {B.~D.}\ \bibnamefont
  {Nelson}}, \ and\ \bibinfo {author} {\bibfnamefont {F.}~\bibnamefont
  {Ruehle}},\ }\href {\doibase 10.1103/PhysRevD.95.043527} {\bibfield
  {journal} {\bibinfo  {journal} {Phys. Rev.}\ }\textbf {\bibinfo {volume}
  {D95}},\ \bibinfo {pages} {043527} (\bibinfo {year} {2017})},\ \Eprint
  {http://arxiv.org/abs/1609.02151} {arXiv:1609.02151 [hep-ph]} \BibitemShut
  {NoStop}%
\bibitem [{\citenamefont {Soni}\ and\ \citenamefont
  {Zhang}(2016)}]{Soni:2016yes}%
  \BibitemOpen
  \bibfield  {author} {\bibinfo {author} {\bibfnamefont {A.}~\bibnamefont
  {Soni}}\ and\ \bibinfo {author} {\bibfnamefont {Y.}~\bibnamefont {Zhang}},\
  }\href@noop {} {\  (\bibinfo {year} {2016})},\ \Eprint
  {http://arxiv.org/abs/1610.06931} {arXiv:1610.06931 [hep-ph]} \BibitemShut
  {NoStop}%
\bibitem [{\citenamefont {da~Rocha}(2017)}]{daRocha:2017cxu}%
  \BibitemOpen
  \bibfield  {author} {\bibinfo {author} {\bibfnamefont {R.}~\bibnamefont
  {da~Rocha}},\ }\href@noop {} {\  (\bibinfo {year} {2017})},\ \Eprint
  {http://arxiv.org/abs/1701.00761} {arXiv:1701.00761 [hep-ph]} \BibitemShut
  {NoStop}%
\bibitem [{\citenamefont {Acharya}\ \emph {et~al.}(2017)\citenamefont
  {Acharya}, \citenamefont {Fairbairn},\ and\ \citenamefont
  {Hardy}}]{Acharya:2017szw}%
  \BibitemOpen
  \bibfield  {author} {\bibinfo {author} {\bibfnamefont {B.~S.}\ \bibnamefont
  {Acharya}}, \bibinfo {author} {\bibfnamefont {M.}~\bibnamefont {Fairbairn}},
  \ and\ \bibinfo {author} {\bibfnamefont {E.}~\bibnamefont {Hardy}},\
  }\href@noop {} {\  (\bibinfo {year} {2017})},\ \Eprint
  {http://arxiv.org/abs/1704.01804} {arXiv:1704.01804 [hep-ph]} \BibitemShut
  {NoStop}%
\bibitem [{\citenamefont {Soni}\ \emph {et~al.}(2017)\citenamefont {Soni},
  \citenamefont {Xiao},\ and\ \citenamefont {Zhang}}]{Soni:2017nlm}%
  \BibitemOpen
  \bibfield  {author} {\bibinfo {author} {\bibfnamefont {A.}~\bibnamefont
  {Soni}}, \bibinfo {author} {\bibfnamefont {H.}~\bibnamefont {Xiao}}, \ and\
  \bibinfo {author} {\bibfnamefont {Y.}~\bibnamefont {Zhang}},\ }\href@noop {}
  {\  (\bibinfo {year} {2017})},\ \Eprint {http://arxiv.org/abs/1704.02347}
  {arXiv:1704.02347 [hep-ph]} \BibitemShut {NoStop}%
\bibitem [{\citenamefont {Hayakawa}(1994)}]{Hayakawa}%
  \BibitemOpen
  \bibfield  {author} {\bibinfo {author} {\bibfnamefont {Y.}~\bibnamefont
  {Hayakawa}},\ }\href@noop {} {\emph {\bibinfo {title} {Degeneration of
  {C}alabi-{Y}au manifold with {W}eil-{P}etersson metric}}}\ (\bibinfo {year}
  {1994})\ p.~\bibinfo {pages} {44},\ \bibinfo {note} {thesis
  (Ph.D.)--University of Maryland, College Park}\BibitemShut {NoStop}%
\bibitem [{\citenamefont {Wang}(1997)}]{Wang}%
  \BibitemOpen
  \bibfield  {author} {\bibinfo {author} {\bibfnamefont {C.-L.}\ \bibnamefont
  {Wang}},\ }\href {\doibase 10.4310/MRL.1997.v4.n1.a14} {\bibfield  {journal}
  {\bibinfo  {journal} {Math. Res. Lett.}\ }\textbf {\bibinfo {volume} {4}},\
  \bibinfo {pages} {157} (\bibinfo {year} {1997})}\BibitemShut {NoStop}%
\bibitem [{\citenamefont {Grassi}(1991)}]{Grassi1991}%
  \BibitemOpen
  \bibfield  {author} {\bibinfo {author} {\bibfnamefont {A.}~\bibnamefont
  {Grassi}},\ }\href {http://eudml.org/doc/164819} {\bibfield  {journal}
  {\bibinfo  {journal} {Mathematische Annalen}\ }\textbf {\bibinfo {volume}
  {290}},\ \bibinfo {pages} {287} (\bibinfo {year} {1991})}\BibitemShut
  {NoStop}%
\bibitem [{\citenamefont {Candelas}\ \emph {et~al.}(2002)\citenamefont
  {Candelas}, \citenamefont {Diaconescu}, \citenamefont {Florea}, \citenamefont
  {Morrison},\ and\ \citenamefont {Rajesh}}]{Candelas:2000nc}%
  \BibitemOpen
  \bibfield  {author} {\bibinfo {author} {\bibfnamefont {P.}~\bibnamefont
  {Candelas}}, \bibinfo {author} {\bibfnamefont {D.-E.}\ \bibnamefont
  {Diaconescu}}, \bibinfo {author} {\bibfnamefont {B.}~\bibnamefont {Florea}},
  \bibinfo {author} {\bibfnamefont {D.~R.}\ \bibnamefont {Morrison}}, \ and\
  \bibinfo {author} {\bibfnamefont {G.}~\bibnamefont {Rajesh}},\ }\href
  {\doibase 10.1088/1126-6708/2002/06/014} {\bibfield  {journal} {\bibinfo
  {journal} {JHEP}\ }\textbf {\bibinfo {volume} {06}},\ \bibinfo {pages} {014}
  (\bibinfo {year} {2002})},\ \Eprint {http://arxiv.org/abs/hep-th/0009228}
  {arXiv:hep-th/0009228 [hep-th]} \BibitemShut {NoStop}%
\bibitem [{\citenamefont {Morrison}()}]{Morrisonunp}%
  \BibitemOpen
  \bibfield  {author} {\bibinfo {author} {\bibfnamefont {D.~R.}\ \bibnamefont
  {Morrison}},\ }\href@noop {} {\bibinfo  {journal} {Unpublished}\
  }\BibitemShut {NoStop}%
\end{thebibliography}%

\end{document}